\pdfoutput=1
\documentclass{amsproc}
\usepackage{a4wide}
\usepackage{bookmark}

\usepackage{}

\newtheorem{theorem}{Theorem}[section]

\theoremstyle{definition}

\theoremstyle{remark}
\newtheorem{remark}[theorem]{Remark}
\newtheorem{corollary}[theorem]{Corollary}
\numberwithin{equation}{section}

\begin{document}
\title[Revisiting calculation of moments of number of comparisons]{Revisiting calculation of moments of number of comparisons\\ used by the randomized quick sort algorithm}
\author[Sumit Kumar Jha]{Sumit Kumar Jha\\ Preprint of the following paper in\\ Discrete Mathematics, Algorithms and Applications:\\ \url{http://www.worldscientific.com/doi/pdf/10.1142/S179383091750001X}}
\address{Center for Security, Theory, and Algorithmic Research\\ 
International Institute of Information Technology, 
Hyderabad, India}
\curraddr{}
\email{kumarjha.sumit@research.iiit.ac.in}
\thanks{}
    \begin{abstract}
We revisit the method of Kirschenhofer, Prodinger and Tichy to calculate the moments of number of comparisons used by the randomized quick sort algorithm. We reemphasize that this approach helps in calculating these quantities with less computation. We also point out that as observed by Knuth this method also gives moments for total path length of a binary search tree built over a random set of $n$ keys. 
\end{abstract}      
\maketitle
\section{Introduction}
\nocite{bworld1} 
\nocite{bworld2}
\nocite{bworld3}
\nocite{bworld4}
Consider the following variant of quick sort algorithm from \cite{Sedgewick2011}: the quick sort algorithm recursively sorts numbers in an array by partitioning it into two smaller and independent subarrays, and thereafter sorting these parts. The partitioning procedure chooses the last element in the array as \emph{pivot} and puts it in its right place where numbers to the left of it are smaller than it, and those to its right are larger than it. \par 
For purposes of this analysis assume that the input array to the quick sort algorithm contains distinct numbers which are randomly ordered. We may assume the input to the algorithm is simply a permutation of $\{1,2,\cdots,n\}$ (if the input array has $n$ elements).\par 
Let $S_{n}$ be the set of all $n!$ permutations of $\{1,2,\cdots,n\}$. Consider a uniform probability distribution on the set $S_{n}$, and define for all $\sigma\in S_{n}$, $C_{n}(\sigma)$ to be the number of comparisons used to sort $\sigma$ by the quick sort algorithm. We wish to calculate mean and variance of $C_{n}$ over the uniform distribution on $S_{n}$.\par 
Our aim here is to obtain following \cite{prodinger52} \cite{bworld}
\begin{theorem}[Knuth \cite{KnuthSort}]
\label{mainthm}
We have
$$\normalfont \text{Mean}(C_{n})=2((n+1)H_{n}-n);$$
and
$$\normalfont \text{Var}(C_{n})=7n^{2}-4(n+1)^{2}H_{n}^{(2)}-2(n+1)H_{n}+13n,$$
over the uniform probability distribution on $S_{n}$. Here we have used the notation $H_{n}=\sum_{k=1}^{n}\frac{1}{k}$ and $H_{n}^{(2)}=\sum_{k=1}^{n}\frac{1}{k^{2}}.$
\end{theorem}
Before proceeding we would like to point out that Hennequin \cite{bworld2} has computed the first five cumulants of the number of comparisons of Quicksort. Also, the variance of the number of comparisons of Quicksort is computed in \cite{bworld1}.
\section{Calculation of mean and variance}
Let $a_{n,s}$ be the number of permutations of $n$ elements requiring a total of $s$ comparisons to sort by the procedure of quicksort.\par 
We start by defining the corresponding \emph{probability generating function}: 
$$G_{n}(z)=\sum_{k\geq 0}\frac{a_{n,k}\, z^{k}}{n!}.$$
\begin{theorem}
For $n\geq 1$
\begin{equation}
\label{eq1}
G_{n}(z)=\frac{z^{n+1}}{n}\sum_{1\leq j \leq n}G_{n-j}(z)G_{j-1}(z),
\end{equation}
and 
\begin{equation}
\label{eq2}
G_{0}(z)=1.
\end{equation}
\end{theorem}
\begin{proof}
The first partitioning stage requires $n+1$ comparisons (for some other variants this might be $n-1$). If the pivot element is $k$th largest, then the sub arrays after partitioning are of sizes $k-1$ and $n-k$. Thus we can write
\begin{equation}
\label{eq3}
a_{n,s}=\sum_{1\leq k \leq n}\binom{n-1}{k-1}\sum_{i+j=s-(n+1)}a_{n-k,i}\,a_{k-1,j}.
\end{equation}
Multiplying equation \eqref{eq3} by $z^{s}$ and dividing by $n!$ we get
\begin{eqnarray*}
\frac{a_{ns}\, z^{s}}{n!}=\sum_{1\leq k \leq n}\frac{z^{s}}{n}\, \sum_{i+j=s-(n+1)}\frac{a_{n-k,i}}{(n-k)!}\cdot \frac{a_{k-1,j}}{(k-1)!}\\
=\sum_{1\leq k \leq n}\frac{z^{s}}{n}\cdot \left\{\text{coefficient of } z^{s-(n+1)} \text{ in } G_{n-k}(z)\cdot G_{k-1}(z)\right\}\\
=\sum_{1\leq k \leq n}\frac{z^{n+1}}{n}\cdot z^{s-(n+1)} \left\{\text{coefficient of } z^{s-(n+1)} \text{ in } G_{n-k}(z)\cdot G_{k-1}(z)\right\}
\end{eqnarray*}
after which summing on $s$ gives us equation \eqref{eq1}.
\end{proof}
We will now consider the \emph{double generating function} $H(z,u)$ defined by
\begin{equation}
\label{eq4}
H(z,u)=\sum_{n\geq 0}G_{n}(z) u^{n}.
\end{equation}
\begin{corollary} We have
\begin{equation}
\label{eq5}
\frac{\partial H(z,u)}{\partial u}=z^{2}\cdot H^{2}(z,zu),
\end{equation}
and 
\begin{equation}
\label{eq6}
H(1,u)=(1-u)^{-1}.
\end{equation}
\begin{proof}
From equation \eqref{eq1} we have
\begin{eqnarray*}
\frac{\partial H(z,u)}{\partial u}=z^{2}\sum_{n\geq 1}(uz)^{n-1}\sum_{1\leq j \leq n}G_{n-j}(z)G_{j-1}(z)\\
=z^{2}\sum_{n\geq 1}(uz)^{n-1} \cdot \left\{\text{coefficient of } (uz)^{n-1} \text{ in } H(z,uz)\cdot H(z,uz) \right\} \\
=z^{2}\cdot H(z,zu)\cdot H(z,zu).
\end{eqnarray*} 
Equation \eqref{eq6} follows from the fact that $G_{n}(1)=1$.
\end{proof}
\end{corollary}
Now we write the $s$th factorial moments $\beta_{s}(n)$ of the random variable with the aid of the probability generating function $G_{n}(z)$:
\begin{equation}
\label{eq7}
\beta_{s}(n)=\left[\frac{d^{s}}{dz^{s}}G_{n}(z)\right]_{z=1}.
\end{equation}
The generating functions $f_{s}(u)$ of $\beta_{s}(n)$ are
\begin{equation}
\label{eq8}
f_{s}(u)=\sum_{n\geq 0}\beta_{s}(n)u^{n}.
\end{equation}
By Taylor's formula and equation \eqref{eq7} we get
\begin{equation}
\label{eq9}
H(z,u)=\sum_{s\geq 0}f_{s}(u)\frac{(z-1)^{s}}{s!}.
\end{equation}
\begin{theorem}
For integer $s\geq 0$ we have
\begin{equation}
\label{eq10}
f'_{s}(u)=s!\cdot \sum_{i+j+k+l+m=s}\frac{a_{i}\cdot f^{(k)}_{j}(u) \cdot f_{l}^{(m)}(u)\cdot u^{k+m} }{j!\cdot k!\cdot  l!\cdot m!},
\end{equation}
where
$$
a_{k}=
\begin{cases}
1&\text{if $k=0$};\\
2&\text{if $k=1$};\\
1&\text{if $k=2$};\\
0&\text{if $k>2$}.
\end{cases}
$$
\end{theorem}
\begin{proof}
Using Taylor's theorem we can write
$$f_{j}(x)=\sum_{k\geq 0}\frac{f_{j}^{(k)}(u)(x-u)^{k}}{k!}$$
which on substituting $x=uz$ gives
\begin{equation}
\label{eq11}
f_{j}(uz)=\sum_{k\geq 0}\frac{f_{j}^{(k)}(u)(z-1)^{k}u^{k}}{k!}.
\end{equation}
\par 
Now substituting equation \eqref{eq9} in equation \eqref{eq5} gives:
\begin{eqnarray*}
\sum_{s\geq 0}f'_{s}(u)\frac{(z-1)^{s}}{s!}=z^{2}\cdot \sum_{p\geq 0}f_{p}(uz)\frac{(z-1)^{p}}{p!} \cdot \sum_{r\geq 0}f_{r}(uz)\frac{(z-1)^{r}}{r!}\\
=\sum_{i\geq 0}a_{i}(z-1)^{i}\cdot \sum_{p\geq 0}\frac{(z-1)^{p}}{p!}\sum_{l\geq 0}\frac{f_{p}^{(l)}(u)(z-1)^{l}u^{l}}{l!}\cdot  \sum_{r\geq 0}\frac{(z-1)^{r}}{r!}\sum_{m\geq 0}\frac{f_{r}^{(m)}(u)(z-1)^{m}u^{m}}{m!}\\
=\sum_{h\geq 0}(z-1)^{h}\sum_{i+j+k+l+m=h}a_{i}\cdot \frac{1}{j!}\cdot \frac{f^{(k)}_{j}(u)\, u^{k}}{k!}\cdot \frac{1}{l!}\cdot \frac{f_{l}^{(m)}(u)\, u^{m}}{m!}
\end{eqnarray*}
where in the second last line we replaced $z^{2}$ by $\sum_{i\geq 0} a_{i}(z-1)^{i}$. Now comparing coefficients on both sides of the equation gives
$$f'_{s}(u)=s!\cdot \sum_{i+j+k+l+m=s}\frac{a_{i}\cdot f^{(k)}_{j}(u) \cdot f_{l}^{(m)}(u)\cdot u^{k+m} }{j!\cdot k!\cdot  l!\cdot m!}.$$
\end{proof}
\begin{remark}
For asymptotic theory of differential equations originating here we recommend reader the paper \cite{bworld4}.
\end{remark}
\begin{corollary}
We have
\begin{equation}
\label{eq12}
f_{0}(u)=(1-u)^{-1},
\end{equation}
\begin{equation}
\label{eq13}
f_{1}(u)=\frac{2}{(1-u)^{2}}\log\frac{1}{1-u},
\end{equation}
\begin{equation}
\label{eq14}
f_{2}(u)=\frac{8\log^{2}(1-u)}{(1-u)^{3}}-\frac{8\log(1-u)}{(1-u)^{3}}-\frac{4\log^{2}(1-u)}{(1-u)^{2}}+\frac{12\log(1-u)}{(1-u)^{2}}+\frac{6}{(1-u)^{3}}-\frac{6}{(1-u)^{2}}.
\end{equation}
\end{corollary}
\begin{proof}
The equation \eqref{eq12} follows from the fact that $\beta_{0}(n)=1$ for all $n\geq 0$.\par 
Setting $s=1$ in equation \eqref{eq10} gives
\begin{eqnarray*}
f'_{1}(u)=a_{1}\cdot f_{0}^{(0)}(u)\cdot f_{0}^{(0)}(u)+f_{0}^{(1)}(u)\cdot f_{0}^{(0)}(u)\cdot u+f_{0}^{(1)}(u)\cdot f_{0}^{(0)}(u)\cdot u\\
+f_{1}^{(0)}(u)\cdot f_{0}^{(0)}(u)
+f_{0}^{(0)}(u)\cdot f_{1}^{(0)}(u)\\
=\frac{2}{(1-u)^{2}}+\frac{u}{(1-u)^{3}}+\frac{u}{(1-u)^{3}}+\frac{f_{1}(u)}{(1-u)}+\frac{f_{1}(u)}{(1-u)},
\end{eqnarray*}
where we used the fact that $f_{0}(u)=(1-u)^{-1}$. The above equation is
\begin{equation}
\label{eq15}
f'_{1}(u)-\frac{2f_{1}(u)}{1-u}=\frac{2u}{(1-u)^{3}}+\frac{2}{(1-u)^{2}}.
\end{equation}
Solving the linear differential equation \eqref{eq15} by multiplying with integrating factor $(1-u)^2$ gives 
$$f_{1}(u)=\frac{2}{(1-u)^{2}}\log\frac{1}{1-u}+f_{1}(0)=\frac{2}{(1-u)^{2}}\log\frac{1}{1-u}.$$
Plugging $s=2$ in \eqref{eq10} and solving the resultant differential equation gives
$$f_{2}(u)=\frac{8\log^{2}(1-u)}{(1-u)^{3}}-\frac{8\log(1-u)}{(1-u)^{3}}-\frac{4\log^{2}(1-u)}{(1-u)^{2}}+\frac{12\log(1-u)}{(1-u)^{2}}+\frac{6}{(1-u)^{3}}-\frac{6}{(1-u)^{2}}.$$\\
\end{proof}
\begin{corollary}
We have
$$\beta_{1}(n)=2((n+1)H_{n}-n),$$
and
$$\beta_{2}(n)=4(n+1)^{2}(H_{n}^{2}-H_{n}^{(2)})+4(n+1)^{2}H_{n}-8(n+1)H_{n}+8nH_{n}-4nH_{n}(5+3n)+11n^{2}+15n.$$
\end{corollary}
\begin{proof}
We use following expansions from \cite{Greene}
$$\frac{1}{(1-u)^{m+1}}\log\left(\frac{1}{1-u}\right)=\sum_{n\geq 0}(H_{n+m}-H_{m})\binom{n+m}{n}u^{n};$$
$$\frac{1}{(1-u)^{m+1}}\log^{2}\left(\frac{1}{1-u}\right)=\sum_{n\geq 0}((H_{n+m}-H_{m})^{2}-(H_{n+m}^{(2)}-H_{m}^{(2)}))\binom{n+m}{n}u^{n},$$
to conclude the assertion.
\end{proof}
\begin{proof}[Proof of Theorem \ref{mainthm}]
We conclude the results after noting
$$\text{Mean}(C_{n})=\beta_{1}(n),$$
$$\text{Var}(C_{n})=\beta_{2}(n)-(\beta_{1}(n))^{2}+\beta_{1}(n).$$
\end{proof}
\section{Similar Partial Differential Functional Equations}
We point out that following two examples from \cite{bworld3} can be analyzed using the method employed here:
\begin{itemize}
\item[1.] Moments of total path length $L_{n}$ of a binary search tree built over a random set of $n$ keys can be extracted from the 
$$\frac{\partial L(z,u)}{\partial z}=L^{2}(zu,u),\quad \frac{\partial L(0,u)}{\partial z}=1,$$
where $L(z,u)=\sum_{n\geq 0}L_{n}(u)z^{n}$ is the bivariate generating function.
\item[2.] A \emph{digital} search tree for which the bivariate generating function $L(z,u)$ satisfies 
$$\frac{\partial L(z,u)}{\partial z}=L^{2}\left(\frac{1}{2}zu,u\right),$$
with $L(z,0)=1.$
\end{itemize}
\bibliographystyle{plain}
\bibliography{sample.bib}
\end{document}